\documentclass{llncs}

\usepackage{todonotes}
\usepackage{array}
\usepackage{float}
\usepackage[only,llbracket,rrbracket,llparenthesis,rrparenthesis]{stmaryrd}
\usepackage{booktabs}
\usepackage{amsmath}
\usepackage{amssymb}
\usepackage{mathpartir}
\usepackage{paralist}
\usepackage{colortbl}

\makeatletter
\def\blfootnote{\gdef\@thefnmark{}\@footnotetext}
\makeatother
\usepackage{macros}

\usepackage[countsame]{coqtheorem}

\title{Verification of PCP-Related \\ Computational Reductions in Coq}
\author{Yannick Forster \and Edith Heiter \and Gert Smolka}
\institute{Saarland University, Saarbr\"ucken, Germany\\
  \path|{forster,heiter,smolka}@ps.uni-saarland.de|}

\newcommand{\N}[1]{#1}

\makeatother

\setBaseUrl{{http://www.ps.uni-saarland.de/extras/PCP/doc/PCP.}}

\begin{document}

\maketitle
\blfootnote{\textcopyright~ Springer International Publishing AG, part of
  Springer Nature 2018 \\
J. Avigad and A. Mahboubi (Eds.): ITP 2018, LNCS 10895, pp. 253–269,
2018.\\
The final authenticated version is available online at \url{https://doi.org/10.1007/978-3-319-94821-8_15}}\vspace{-5mm}
\begin{abstract}
  We formally verify several computational reductions concerning the
  Post correspondence problem (PCP) using the proof assistant Coq.
  Our verification includes a reduction of the halting problem for
  Turing machines to string rewriting, a reduction of string rewriting
  to PCP, and reductions of PCP to the intersection problem and the
  palindrome problem for context-free grammars.

\end{abstract}\enlargethispage{1.5cm}%
\section{Introduction}

A problem $P$ can be shown undecidable
by giving an undecidable problem $Q$ and
a computable function reducing~$Q$ to~$P$.
There are well known reductions of the
halting problem for Turing machines (TM)
to the Post correspondence problem (PCP),
and of PCP to the intersection problem
for context-free grammars (CFI).
We study these reductions in the formal setting
of Coq's type theory~\cite{Coq} with the goal of
providing elegant correctness proofs.

Given that the reduction of TM to PCP
appears in textbooks~\cite{hopcroft,davis,sipser}
and in the  standard
curriculum for theoretical computer science,
one would expect that rigorous
correctness proofs can be found in the literature.
To our surprise, this is not the case.
Missing is the formulation of
the inductive invariants enabling the necessary inductive proofs to go through.
Speaking with the analogue of imperative programs,
the correctness arguments in the literature
argue about the correctness of programs with loops
without stating and verifying loop invariants.

By inductive invariants we mean statements that
are shown inductively and that generalise
the obvious correctness statements one starts with.
Every substantial formal correctness proof
will involve the construction of suitable
inductive invariants.
Often it takes ingenuity to generalise
a given correctness claim to one or several
inductive invariants that can be shown inductively.

It took some effort to come up with the
missing inductive invariants for the
reductions leading from TM to PCP.
Once we had the inductive invariants,
we had rigorous and transparent proofs
explaining the correctness of the reductions
in a more satisfactory way than
the correctness arguments we found in the literature.

Reduction of problems is transitive.
Given a reduction $P\preceq Q$ and a reduction $Q\preceq R$,
we have a reduction $P\preceq R$.
This way, complex reductions
can be factorised into simpler reductions.
Following ideas in the literature,
we will establish the reduction chain
$$
\text{TM}\preceq
\text{SRH}\preceq
\text{SR}\preceq
\text{MPCP}\preceq
\text{PCP}
$$
where TM is the halting problem of single-tape Turing machines, SRH is a generalisation of
the halting problem for Turing machines,
SR is the string rewriting problem,
and MPCP is a modified version of PCP fixing a first card.
The most interesting steps are
$\text{SR}\preceq\text{MPCP}$ and
$\text{MPCP}\preceq\text{PCP}$.

We also consider the intersection problem (CFI) and
the palindrome problem (CFP) for a class of
linear context-free grammars we call Post grammars.
CFP asks whether a Post grammar generates a palindrome,
and CFI asks whether for two Post grammars
there exists a string generated by both grammars.
We will verify reductions $\text{PCP}\preceq\text{CFI}$
and ${\text{PCP}\preceq\text{CFP}}$,
thus showing that CFP and CFI are both undecidable.

Coq's type theory provides an ideal setting
for the formalisation and verification of
the reductions mentioned.
The fact that all functions in Coq are total and computable
makes the notion of computable reductions straightforward.

The correctness arguments coming with our approach
are inherently constructive,
which is verified by the underlying constructive type theory.
The main inductive data types we use are numbers and lists,
which conveniently provide for
the representation of strings, rewriting systems,
Post correspondence problems, and Post grammars.

The paper is accompanied by a Coq development
covering all results of this paper.
The definitions and statements in the paper are
hyperlinked with their formalisations
in the HTML presentation of the Coq development
at \url{http://www.ps.uni-saarland.de/extras/PCP}.

\subsection*{Organisation}

We start with the necessary formal definitions
covering all reductions we consider in Section~2.
We then present each of the six reductions
and conclude with a discussion of the
design choices underlying our formalisations.
Sections 3 to 8 on the reductions are independent
and can be read in any order.

We only give definitions for the problems and do not discuss the
underlying intuitions, because all problems are covered in a typical
introduction to theoretical computer science and the interested reader
can refer to various textbooks providing good intuitions,
e.g.~\cite{hopcroft,sipser,davis}.

\subsection*{Contribution}

Our reduction functions follow the ideas in the literature. The main
contributions of the paper are the formal correctness proofs for the
reduction functions.  Here some ingenuity and considerable elaboration
of the informal arguments in the literature were needed.  As one would
expect, the formal proofs heavily rely on inductive techniques.  In
contrast, the informal proof sketches in the literature do not
introduce the necessary inductions (in fact, they don't even mention
inductive proofs).  To the best of our knowledge, the present paper is
the first paper providing formal correctness proofs for basic
reductions to and from PCP.

\section{Definitions}
\setCoqFilename{Definitions}

Formalising problems and computable reductions
in constructive type theory is straightforward.
A \emph{problem} consists of a type $X$
and a unary predicate $p$ on $X$,
and a \emph{reduction} of $(X,p)$ to $(Y,q)$
is a function $f:X\to Y$ such that
$\forall x.~px\toot q(fx)$.
Note that the usual requirement
that $f$ is total and computable
can be dropped since it is satisfied
by every function in a constructive type theory.
We write \emph{$p\preceq q$} and say
that \emph{$p$ reduces to $q$}
if a reduction of $(X,p)$ to $(Y,q)$ exists.

\begin{fact}[][reduces_transitive]
  If $p\preceq q$ and $q\preceq r$, then $p\preceq r$.
\end{fact}

The basic inductive data structures we use are
\emph{numbers} ($n::=0\mid Sn$) and
\emph{lists} ($L::=\nil\mid s::L$).
We write
\emph{$L_1\con L_2$} for the \emph{concatenation} of two lists,
\emph{$\ol L$} for the reversal of a list,
\emph{$[\,fs\mid s\in A\,]$} for a map over a list,
and \emph{$[\,fs\mid s\in A\land ps\,]$} for
a map and filter over a list.
Moreover, we write
\emph{$s\in L$} if $s$ is a member of $L$, and
\emph{$L_1\incl L_2$} if every member of~$L_1$ is a member of $L_2$.

A \emph{string} is a list of symbols,
and a \emph{symbol} is a number.
The letters $x$, $y$, $z$, $u$, and $v$ range over strings,
and the letters $a$, $b$, $c$ range over symbols.
We write \emph{$xy$} for $x\con y$ and \emph{$ax$} for $a::x$.
We use \emph{$\epsilon$} to denote the empty string.
A \emph{palindrome} is a string $x$ such that $x=\ol x$.

\begin{fact}\label{fact-CFP-rev}
  $\ol{xy}=\ol y\,\ol x$ and $\,\ol{\ol x{}}=x$.
\end{fact}

\begin{fact}[][list_prefix_inv]\label{fact-prelim-split}
  If $xay=uav$, $a\notin x$, and $a \notin u$, then $x=u$ and $y=v$.
\end{fact}
\begin{proof}
  By induction on $x$.
\qed\end{proof}

A \emph{card $x/y$} or a \emph{rule $x/y$}
a is a pair $(x,y)$ of two strings.
When we call $x/y$ a card we see
$x$ as the upper and  $y$ as the lower string of the card.
When we call $x/y$ a rule we see
$x$ as the left and  $y$ as the right side of the rule.

The letters $A$, $B$, $C$, $P$, $R$ range over list of cards or rules.

\subsection{Post Correspondence Problem}

A \emph{stack} is a list of cards.
The \coqlink[tau1]{\emph{upper trace $A^1$}} and
the \coqlink[tau2]{\emph{lower trace $A^2$}}
of a stack $A$ are strings defined as follows:
\begin{align*}
  \nil^1&:=\epsilon&\nil^2&:=\epsilon\\
  {(x/y::A)}^1&:=x(A^1)&{(x/y::A)}^2&:=y(A^2)
\end{align*}
Note that $A^1$ is the
concatenation of the upper strings of the cards in $A$,
and that~$A^2$ is the
concatenation of the lower strings of the cards in $A$.
We say that a stack $A$ \emph{matches} if $A^1=A^2$ and a \emph{match} is a matching stack.
An example for a match is the list
$A=[\epsilon/ab,~a/c,~bc/\epsilon]$,
which satisfies $A^1=A^2=abc$.

We can now define the predicate
for the \emph{Post correspondence problem}:
\begin{align*}
  \N{\M{PCP}\,(P)}&~:=~\exists A\incl P.~A\neq\nil\land A^1=A^2
\end{align*}
Note that $\M{PCP}\,(P)$ holds iff there exists
a nonempty match $A\incl P$.
We then say that~$A$ is a \emph{solution} of $P$.
For instance,
$$P=[a/\epsilon,\,b/a,\,\epsilon/bb]$$
is solved by the match
$$A=[\epsilon/bb,\,b/a, \,b/a,\,a/\epsilon,\,a/\epsilon].$$

While it is essential that $A$ is a list
providing for order and duplicates,
$P$ may be thought of as a finite set of cards.

We now define the predicate
for the \emph{modified Post correspondence problem}:
\begin{align*}
  \N{\M{MPCP}\,(x/y,P)}
  &~:=~\exists A\incl x/y::P.~xA^1=yA^2
\end{align*}
Informally, $\M{MPCP}\,(x/y,P)$ is like $\M{PCP}\,(x/y::P)$
with the additional constraint
that the solution for $x/y::P $
starts with the first card $x/y$.

Note that in contrary to most text books we leave open whether $x/y$ is an
element of $P$ and instead choose $A$ as subset of $x/y :: P$.
While this might first seem more complicated, it actually
eases formalisation.
Including $x/y$ into $P$ would require $\M{MPCP}$ to be a predicate on
arguments of the form $(P, x/y, H : x/y \in P)$, i.e. dependent
pairs containing a proof.

\subsection{String Rewriting}

Given a list $R$ of rules,
we define \emph{string rewriting}
with two inductive predicates
\coqlink[rew]{\emph{$x\succ_Ry$}} and \coqlink[rewt]{\emph{$x\succ^*_Ry$}}:
\begin{mathpar}
  \inferrule*
  {x/y\in R}
  {uxv\succ_Ruyv}
  \and
  \inferrule*~{z\succ^*_Rz}
  \and
  \inferrule*
  {x\succ_Ry\\y\succ^*_Rz}
  {x\succ^*_Rz}
\end{mathpar}
Note that
$\succ^*_R$ is the reflexive transitive closure of $\succ_R$,
and that $x\succ_Ry$ says that
$y$ can be obtained from $x$
with a single rewriting step
using a rule in $R$.

\begin{fact}\label{fact-SR-basics}
  The following hold:
  \begin{enumerate}
  \coqitem[PreOrder_rewt] If $x\succ^*_Ry$ and $y\succ^*_Rz$, then $x\succ^*_Rz$.
  \coqitem[rewt_app] If $x\succ^*_Ry$, then $ux\succ^*_Ruy$.
  \coqitem[rewt_subset] If $x\succ^*_Ry$ and $R\incl P$, then $x\succ^*_Py$.
  \end{enumerate}
\end{fact}
\begin{proof}
  By induction on $x\succ^*_Ry$.
\qed\end{proof}

Note that the induction lemma for string rewriting can be stated as
$$\forall z.~P z\to (\forall x y.~x \succ_R y \to P y \to P x) \to \forall x.~x
\succ^*_R z \to P x.$$
This is stronger than the lemma Coq infers, because of the
quantification over $z$ on the outside. The quantification
is crucial for many proofs that do induction on derivations $x \succ_R
z$, and we use the lemma throughout the paper without explicitly
mentioning it.

We define the predicates for
the \emph{string rewriting problem}
and the \emph{generalised halting problem} as follows:
\begin{align*}
  \N{\M{SR}\,(R,x,y)}&~:=~x\succ^*_Ry\\
  \N{\M{SRH}\,(R,x,a)}&~:=~\exists y.~x\succ^*_Ry\land a\in y
\end{align*}

We call the second problem \emph{generalised halting problem}, because
it covers the halting problem for deterministic single-tape Turing machines, but
also the halting problems for nondeterministic machines or for
more exotic machines that e.g. have a one-way infinite tape or can
read multiple symbols at a time.

We postpone the definition of Turing machines and of the halting problem TM to section~\ref{sec:TM-SRH}.

\subsection{Post Grammars}

A \emph{Post grammar} is a pair $(R,a)$
of a list $R$ of rules and a symbol~$a$.
Informally, a Post grammar $(R,a)$ is a special case of a
context-free grammar with a single nonterminal~$S$
and two rules $S\to xSy$ and $S\to xay$
for every rule $x/y\in R$,
where $S\neq a$ and~$S$ does not occur in $R$.
We define the \coqlink[sigma]{\emph{projection $\sigma_aA$}}
of a list of rules $A$ with a symbol~$a$ as follows:
\begin{align*}
  \sigma_a\nil&:=a\\
  \sigma_a(x/y::A)&:=x(\sigma_aA)y
\end{align*}
We say that a Post grammar $(R,a)$
\emph{generates} a string $u$
if there exists a nonempty list $A\incl R$
such that $\sigma_aA=u$.
We then say that~$A$ is a \emph{derivation of $u$}
in~$(R,a)$.

We can now define the predicates for the problems CFP and CFI:
\begin{align*}
  \N{\M{CFP}\,(R,a)}
  &~:=~\exists A\incl R.~
    A\neq\nil\land\sigma_aA=\ol{\sigma_aA}\\
  \N{\M{CFI}\,(R_1,R_2,a)}
  &~:=~\exists A_1\incl R_1~\exists A_2\incl R_2.\\
  &\hskip10mm
    A_1\neq\nil\land A_2\neq\nil\land\sigma_aA_1=\sigma_aA_2
\end{align*}
Informally,
$\M{CFP}\,(R,a)$ holds iff
the grammar $(R,a)$ generates a palindrome,
and $\M{CFI}\,(R_1,R_2,a)$ holds iff
there exists a string that is generated by
both grammars $(R_1,a)$ and $(R_2,a)$.
\setCoqFilename{Post_CFG}
Note that as Post grammars are special cases of context-free grammars,
the reduction of PCP to CFG and CFI can be trivially extended to
reductions to the respective problems for context-free grammars.
We prove this formally \coqlink[reduce_grammars]{in the accompanying Coq development}.

\subsection{Alphabets}

For some proofs it will be convenient to fix
a finite set of symbols.
We represent such sets as lists and speak
of \emph{alphabets}.
The letter~\emph{$\Sigma$} ranges over alphabets.
We say that an alphabet $\Sigma$
\emph{covers} a string, card, or stack
if $\Sigma$ contains every symbol
occurring in the string, card, or stack.
We may write \emph{$x\incl\Sigma$} to
say that $\Sigma$ covers~$x$
since both $x$ and $\Sigma$ are lists of symbols.

\subsection{Freshness}
\setCoqFilename{Definitions}
\newcommand{\fresh}{\M{fresh}\ }
At several points we will need to pick fresh symbols from an alphabet.
Because we model symbols as natural numbers, a very simple definition
of freshness suffices.
We define a function $\fresh$such that $\fresh \Sigma \not \in \Sigma$ for an alphabet $\Sigma$ as follows:
\begin{align*}
  \fresh \nil &= 0 \\
  \fresh (a :: \Sigma) &= 1 + a + \fresh \Sigma
\end{align*}

$\fresh$has the following characteristic property:

\begin{lemma}[][fresh_spec']
  For all $a \in \Sigma$, $\fresh \Sigma > a$.
\end{lemma}
\begin{proof}
  By induction on $\Sigma$, with $a$ generalised.
\qed\end{proof}

The property is most useful when exploited in the following way:
\begin{corollary}[][fresh_spec]
  For all $a \in \Sigma$, $\fresh \Sigma \neq a$.
\end{corollary}

An alternative approach to this is to formalise alphabets explicitly
as types $\Sigma$.
This has the advantage that arbitrarily many fresh symbols can
be introduced simultaneously using definitions like $\Gamma := \Sigma + X$, and
symbols in $\Gamma$ stemming from $\Sigma$ can easily be shown different from fresh symbols
stemming from $X$ by inversion.
However, this means that strings $x : \Sigma^*$ have to be explicitly
embedded pointwise when used as strings of type $\Gamma^*$, which
complicates proofs.

In general, both approaches have benefits and tradeoffs. Whenever
proofs rely heavily on inversion (as e.g. our proofs in Section 8),
the alternative approach is favorable. If proofs need the construction
of many strings, as most of our proofs do, modelling symbols as
natural numbers shortens proofs.

\section{SRH to SR}
\setCoqFilename{SRH_SR}

We show that SRH (the generalised halting problem) reduces to SR (string rewriting).
We start with the definition of the reduction function.
Let $R$, $x_0$, and $a_0$ be given.

We fix an alphabet $\Sigma$ covering $R$, $x_0$, and $a_0$.
We now add rules to $R$ that allow $x \succ^*_R a_0$ if $a_0 \in x$.
$$P:=R\cup\,
\mset{aa_0/a_0}{a\in\Sigma}\cup\,
\mset{a_0a/a_0}{a\in\Sigma}$$

\begin{lemma}[][x_rewt_a0]\label{lem-SRH-1}
  If $a_0\in x\incl\Sigma$, then $x\succ^*_Pa_0$.
\end{lemma}
\begin{proof}
  For all $y \subseteq \Sigma$, $a_0y \succ_P^* a_0$ and $ya_0 \succ_P^* a_0$ follow by induction on $y$.
  The claim now follows with
  Fact~\ref{fact-SR-basics}~(1,2).
\qed\end{proof}

\begin{lemma}[][equi]\label{lem-SRH-equi}
  $\M{SRH}\,(R,x_0,a_0)\toot\M{SR}\,(P,x_0,a_0)$.
\end{lemma}
\begin{proof}
  Let $x_0\succ^*_R y$ and $a_0\in y$.
  Then $y\succ^*_Pa_0$ by Lemma~\ref{lem-SRH-1}.
  Moreover, $x_0\succ^*_P y$
  by Fact~\ref{fact-SR-basics}~(3).
  Thus $x_0\succ^*_Pa_0$ by Fact~\ref{fact-SR-basics}~(1).

  Let $x_0\succ^*_Pa_0$.
  By induction on $x_0\succ^*_Pa_0$ it follows
  that there exists $y$ such that
  $x_0\succ^*_R y$ and $a_0\in y$.
\qed\end{proof}

\begin{theorem}[][reduction]
  SRH reduces to SR.
\end{theorem}
\begin{proof}
  Follows with Lemma~\ref{lem-SRH-equi}.
\qed\end{proof}

\section{SR to MPCP}
\setCoqFilename{SR_MPCP}

We show that SR (string rewriting) reduces to MPCP (the modified Post
correspondence problem).
We start with the definition of the reduction function.

Let $R$, $x_0$ and $y_0$ be given.
We fix an alphabet $\Sigma$ covering $R$, $x_0$, and $y_0$.
We also fix two symbols $\$,\#\notin\Sigma$
and define:
\begin{align*}
  \coqlink[d]{d}&~:=~\$\,/\,\$x_0\#\\
  \coqlink[e]{e}&~:=~y_0\#\$\,/\,\$\\
  \coqlink[P]{P}&~:=~\set{d,e}\cup R \cup\,\set{\#/\#}\cup\,\mset{a/a}{a\in\Sigma}
\end{align*}

The idea of the reduction is as follows:
Assume $\Sigma = [a,b,c]$ and rules $bc / a$ and $aa/b$ in
$R$. Then $abc \succ_R aa \succ_R b$ and we have
$d = \$ / \$ abc \#$, $e = b\#\$/\$$, and
$P = [d, e, bc/a, aa/b, \dots, a/a, b/b,c/c]$, omitting possibe further rules in $R$.
Written suggestively, the following stack matches:

\[  \scbox{\$}{\$ abc \#} \scbox{a}{a} \scbox{bc}{a} \scbox{\#}{\#}
  \scbox{aa}{b} \scbox{\#}{\#} \scbox{b\#\$}{\$} \]

And, vice versa, every matching stack starting with $d$ will yield
a derivation of $abc \succ_R^* b$.

We now go back to the general case and state the correctness lemma for the reduction function.

\begin{lemma}[][SR_MPCP_cor]
  \label{lem-SRMPCP-cor}
  $x_0\succ^*_R y_0$ if and only if
  there exists a stack $A\incl P$ such that
  $d::A$ matches.
\end{lemma}
From this lemma we immediately obtain the
reduction theorem (Theorem~\ref{theo-SRMPCP}).
The proof of the lemma consists of two \emph{translation lemmas}:
Lemma~\ref{lem-SRMPCP-trans1} and
Lemma~\ref{lem-SRMPCP-trans2}.
The translation lemmas generalise
the two directions of Lemma~\ref{lem-SRMPCP-cor}
such that they can be shown with canonical inductions.

\begin{lemma}[][SR_MPCP]\label{lem-SRMPCP-trans1}
  Let $x\incl\Sigma$ and $x\succ^*_R y_0$.
  Then there exists $A\incl P$ such that $A^1=x\#A^2$.
\end{lemma}
\begin{proof}
  By induction on $x\succ^*_R y_0$.
  In the first case, $x = y_0$ and $[e]^1 = x \# [e]^2$.
  In the second case, $x \succ y$ and $y \succ^* y_0$.
  By induction hypothesis there is $A \incl P$ such that $A^1 = y
  \# A^2$.
  Let $x = (a_1 \dots a_n)u (b_1 \dots b_n)$ and $y = (a_1
  \dots a_n) v  (b_1 \dots b_n)$ for $u/v \in R$.
  We define $B := (a_1/a_1) \dots (a_n / a_n) ::  (u/v) :: (b_1 / b_1) \dots (b_n /
  b_n) :: (\#/\#)  :: A$.
  Now $B^1 = x \# A^1 = x \# y \# A^2 = x \# B^2$.
\qed\end{proof}

\begin{lemma}[][MPCP_SR]\label{lem-SRMPCP-trans2}
  Let $A\incl P$,
  $A^1=x\#yA^2$,
  and $x,y\incl\Sigma$.
  Then $yx\succ^*_Ry_0$.
\end{lemma}
\begin{proof}
  By induction on $A$ with $x$ and $y$ generalised.
  We do all cases in detail:
  \begin{itemize}
  \item The cases where $A = \nil$ or $A = d :: B$ are contradictory.
  \item Let $A = e :: B$. By assumption, $y_0 \# \$ B^1 = x \# y \$ B^2$.
    Then $x = y_0$, $y = \epsilon$ and $yx = y_0 \succ_R^* y_0$.
  \item Let $A = u/v :: B$ for $u/v \in R$.
    Because $\#$ is not in $u$ and by assumption $u B^1 = x \# y v B^2$, $x = u \app x'$.
    And $yx = yux' \succ yvx' \succ^* y_0$ by induction hypothesis.
  \item Let $A = \#/\# :: B$. By assumption, $\#B^1 = x\#y\#B^2$. Then $x =
    \epsilon$ and we have $B^1 = y\# \epsilon B^2$. By induction
    hypothesis, this yields $yx = \epsilon y \succ_R^* y_0$ as needed.
  \item Let $A = a/a :: B$ for $a \in \Sigma$ and assume $a B^1 = x \#
    y a B^2$. Then $x = a x'$ and $B^1 = x'\#yaB^2$. By induction
    hypothesis, this yields $yx = yax' \succ_R^* y_0$ as needed.
  \end{itemize}
\qed\end{proof}

\begin{theorem}[][reduction]\label{theo-SRMPCP}
  SR reduces to MPCP.
\end{theorem}
\begin{proof}
  Follows with Lemma~\ref{lem-SRMPCP-cor}.
\qed\end{proof}

The translation lemmas formulate
what we call the \emph{inductive invariants}
of the reduction function.
The challenge of proving the correctness
of the reduction function is finding
strong enough inductive invariants
that can be verified with canonical inductions.

\section{MPCP to PCP}
\setCoqFilename{MPCP_PCP}

We show that MPCP (modified PCP) reduces to PCP.

The idea of the reduction is that for a stack $A = [x_1/y_1, \dots,
x_n, y_n]$ and a first card $x_0/y_0$ where $x_i = a_i^0 \dots
a_i^{m_i}$ and $y_i = b_i^0 \dots b_i^{m'_i}$ we
have \enlargethispage{1mm}
\begin{align*}
&(a_0^0 \dots a_0^{m_0}) (a_1^0 \dots a_1^{m_1}) \dots (a_n^0
\dots a_n^{m_n}) \\ = &(b_0^0 \dots b_0^{m'_0}) (b_1^0 \dots b_1^{m'_1}) \dots (b_n^0
\dots b_n^{m'_n})
\end{align*}

if and only if we have
\begin{align*}
&\$ ( \# a_0^0 \dots \# a_0^{m_0}) (\# a_1^0 \dots \# a_1^{m_1})
  \dots (\# a_n^0 \dots \# a_n^{m_n}) \# \$ \\ = &\$ \# (b_0^0 \# \dots b_0^{m'_0}
  \#) (b_1^0 \# \dots b_1^{m'_1} \#) \dots (b_n^0 \#
\dots b_n^{m'_n} \#) \$.
\end{align*}

The reduction function implements this idea by constructing a
dedicated first and a dedicated last card and by inserting $\#$-symbols
into the MPCP cards:

Let $x_0/y_0$ and $R$ be given.
We fix an alphabet $\Sigma$ covering $x_0/y_0$ and $R$.
We also fix two symbols $\$,\#\notin\Sigma$.
We define two functions \coqlink[hash_L]{\emph{$^\#x$}} and \coqlink[hash_R]{\emph{$x^\#$}}
inserting the symbol $\#$ before and after
every symbol of a string $x$:
\begin{align*}
  ^\#\epsilon&:=\epsilon&\epsilon ^\#&:=\epsilon\\
  ^\#(ax)&:=\#a(^\#x)&(ax) ^\#&:=a\#(x^\#)
\end{align*}
We define:
\begin{align*}
  \coqlink[d]{d}&~:=~\$(^\#x_0)\,/\,\$\#(y_0^\#)\\
  \coqlink[e]{e}&~:=~\#\$\,/\,\$\\
  \coqlink[P]{P}&~:=~\set{d,e}\cup\,
     \mset{\,^\#x\,/\,y^\#}{x/y\in x_0/y_0 :: R\land (x/y) \neq (\epsilon/\epsilon)}
\end{align*}
We now state the correctness lemma for the reduction function.

\begin{lemma}[][MPCP_PCP_cor]
  \label{lem-MPCP-PCP-cor}
  There exists  a stack $A\incl x_0/y_0::R$
  such that $x_0A^1=y_0A^2$
  if and only if
  there exists a nonempty stack $B\incl P$ such that $B^1=B^2$.
\end{lemma}
From this lemma we immediately obtain the
desired reduction theorem (Theorem~\ref{theo-MPCP-PCP}).
The proof of the lemma consists of two translation lemmas
(Lemmas~\ref{lem-MPCP-PCP-trans1} and~\ref{lem-MPCP-PCP-trans2})
and a further auxiliary lemma (Lemma~\ref{lem-MPCP-PCP-d-first}).

\begin{lemma}[][match_start]\label{lem-MPCP-PCP-d-first}
  Every nonempty match $B\incl P$ starts with $d$.
\end{lemma}
\begin{proof}
  Let $B$ be a nonempty match $B\incl P$.
  Then $e$ cannot be the first card of $B$
  since the upper string and lower string of~$e$
  start with different symbols.
  For the same reason ${^\#x}\,/\,y^\#$
  cannot be the first card of $B$
  if $x/y\in R$ and both $x$ and $y$ are nonempty.

  Consider $\epsilon/ay\in R$.
  Then $\epsilon/(ay)^\#$ cannot be the first card of $B$
  since no card of~$P$ has an upper string starting with $a$.

  Consider $ax/\epsilon\in R$.
  Then $^\#(ay)/\epsilon$ cannot be the first card of $B$
  since no card of~$P$ has a lower string starting with $\#$.
\qed\end{proof}

For the proofs of the translation lemmas
we need a few facts about $^\#x$ and $x^\#$.

\begin{lemma}\label{lem-MPCP-PCP-sharp}
  The following hold:
  \begin{enumerate}
    \coqitem[hash_swap] $(^\#x)\#=\#(x^\#)$.
    \coqitem[hash_L_app] $^\#(xy)=(^\#x)(^\#y)$.
    \coqitem[hash_R_app] $(xy)^\#=(x^\#)(y^\#)$.
    \coqitem[hash_L_diff] $^\#x\neq\#(y^\#)$.
    \coqitem[hash_R_inv] $x^\#=y^\#\to x=y$.
  \end{enumerate}

\end{lemma}
\begin{proof}
  By induction on $x$.
\qed\end{proof}

\begin{lemma}[][MPCP_PCP]\label{lem-MPCP-PCP-trans1}
  Let $A\incl x_0/y_0::R$ and $xA^1=yA^2$.
  Then there exists a stack $B\incl P$
  such that $(^\#x)B^1=\#(y^\#)B^2$.
\end{lemma}
\begin{proof}
  By induction on $A$ with $x$ and $y$ generalised.
  The case for $A = \nil$ follows from Lemma~\ref{lem-MPCP-PCP-sharp}
  (1) by choosing $[e]$.

  For the other case, let $A = x'/y' :: A'$.
  Then by assumption $x x' A'^1 = y y' A'^2$. 
  And thus by induction hypothesis there exists $B \subseteq P$ such that
  $^\#(x x') B^1 = \# (y y')^\# B^2$.
  By  Lemma~\ref{lem-MPCP-PCP-sharp} (2) and (3),
  $(^\#x) (^\#x') B^1 = \# (y ^\#) ({y'}^\#) B^2$.

  If $(x'/y') \neq (\epsilon/\epsilon)$, then choosing $^\# x' / {{y'}
    ^\#} :: B \subseteq P$ works.
  Otherwise, $B \subseteq P$ works.
\qed\end{proof}

\begin{lemma}[][PCP_MPCP]\label{lem-MPCP-PCP-trans2}
  Let $B\incl P$ such that
  $(^\#x)B^1=\#(y^\#)B^2$ and $x,y\incl\Sigma$.
  Then there exists a stack $A\incl x_0/y_0::R$ such that $xA^1=yA^2$.
\end{lemma}
\begin{proof}
  By induction on $B$.
  The cases $B = \nil$ and $B = d :: B'$ yield contradictions using
  Lemma~\ref{lem-MPCP-PCP-sharp} (4).
  For $B = e :: B'$, choosing $A = \nil$ works by
  Lemma~\ref{lem-MPCP-PCP-sharp}~(5).

  The interesting case is $B = ^\# x' / {y'} ^\# :: B'$ for $x'/ {y'}
  \in x_0/y_0 :: R$ with $ (x'/y') \neq (\epsilon/\epsilon)$. 
  By assumption and Lemma~\ref{lem-MPCP-PCP-sharp} (2) and (3) we know
  that $^\# (x {x'}) B'^1 = \# (y {y'})^\# B'^2$.
  Now by induction hypothesis, where all premises follow easily, there
  is $A \subseteq x_0/y_0 :: R$ with $x  x'  A^1 = y
  y'  A^2$ and thus $x'/y' :: A$ works.
  \qed\end{proof}

\begin{theorem}[][reduction]\label{theo-MPCP-PCP}
  MPCP reduces to PCP.
\end{theorem}
\begin{proof}
  Follows with Lemma~\ref{lem-MPCP-PCP-cor}.
\qed\end{proof}

\section{PCP to CFP}
\setCoqFilename{PCP_CFP}

We show that PCP reduces to CFP (the palindrome problem for Post grammars).

Let $\#$ be a symbol.

\begin{fact}\label{fact-CFP-palin}
  Let $\#\notin x,y$.
  Then $x\# y$ is a palindrome iff $y=\ol x$.
\end{fact}
\begin{proof}
  Follows with Facts~\ref{fact-CFP-rev} and~\ref{fact-prelim-split}.
\qed\end{proof}

There is an obvious connection between
matching stacks and palindromes:
A stack
$$[x_1/y_1\cld x_n/y_n]$$
matches if and only if the string
$$x_1\cdots x_n\#\,\ol{y_n}\cdots\ol{y_1}$$
is a palindrome,
provided the symbol $\#$ does not appear in the stack
(follows with
Facts~\ref{fact-CFP-rev} and~\ref{fact-CFP-palin}
using $\ol{y_n}\cdots\ol{y_1}=\ol{y_1\cdots y_n}\,$).
Moreover, strings of the form
$x_1\cdots x_n\#\,\ol{y_n}\cdots\ol{y_1}\,$
with $n\ge1$
may be generated by a Post grammar
having a rule $x/\,\ol y\,$
for every card $x/y$ in the stack.
The observations yield a reduction of PCP to CFP.

We formalise the observations with a function
$$\gamma A:=\mset{x/\ol y}{x/y\in A}.$$

\begin{lemma}[][sigma_gamma]\label{lem-CFP-0}
  $\sigma_\#(\gamma A)=A^1\,\#\,\ol{A^2}$.
\end{lemma}
\begin{proof}
  By induction on $A$ using Fact~\ref{fact-CFP-rev}. 
\qed\end{proof}

\begin{lemma}[][tau_eq_iff]\label{lem-CFP1}
  Let $A$ be a stack and
  $\#$ be a symbol not occurring in $A$.
  Then $A$ is a match if and only if
  $\sigma_\#(\gamma A)$ is a palindrome.
\end{lemma}
\begin{proof}
  Follows with
  Lemma~\ref{lem-CFP-0} and
  Facts~\ref{fact-CFP-palin} and~\ref{fact-CFP-rev}.
\qed\end{proof}

\begin{lemma}[][gamma_invol]\label{lem-CFP2}
  $\gamma(\gamma A)=A$ and
  $A\incl\gamma B\to\gamma A\incl B$.
\end{lemma}
\begin{proof}
  By induction on $A$ using Fact~\ref{fact-CFP-rev}. 
\qed\end{proof}

\begin{theorem}[][PCP_CFP]
  PCP reduces to CFP.
\end{theorem}
\begin{proof}
  Let $P$ be a list of cards.
  We fix a symbol $\#$ that is not in $P$
  and show ${\M{PCP}\,(P)\toot\M{CFP}\,(\gamma P,\#)}$.

  Let $A\incl P$ be a nonempty match.
  It suffices to show that $\gamma A\incl\gamma P$
  and $\sigma_\#(\gamma A)$ is a palindrome.
  The first claim follows with Lemma~\ref{lem-CFP2},
  and the second claim follows with Lemma~\ref{lem-CFP1}.

  Let $B\incl\gamma P$ be a nonempty stack
  such that $\sigma_\#B$ is a palindrome.
  By Lemma~\ref{lem-CFP2} we have $\gamma B\incl P$
  and  $B=\gamma(\gamma B)$.
  Since $\gamma B$ matches by Lemma~\ref{lem-CFP1},
  we have $\M{PCP}\,(P)$.
\qed\end{proof}

\section{PCP to CFI}
\setCoqFilename{PCP_CFI}

We show that PCP reduces to CFI (the intersection problem for Post grammars).
The basic idea is that a stack
$A=[x_1/y_1,\ldots,x_n/y_n]$ with $n\ge1$ matches if and only if the
string
$$x_1\cdots x_n\#\,x_n\#y_n\#\cdots\# x_1\#y_1\#$$
equals the string
$$y_1\cdots y_n\#\,x_n\#y_n\#\cdots\# x_1\#y_1\#$$
provided the symbol $\#$ does not occur in $A$.
Moreover, strings of these forms can be generated
by the Post grammars
$([\,x/x\#y\#\mid x/y\in A\,],\,\#)$ and
$([\,y/x\#y\#\mid x/y\in A\,],\,\#)$, respectively.

We fix a symbol $\#$ and formalise
the observations with two functions
\begin{align*}
  \coqlink[gamma1]{\N{\gamma_1A}}&:=\mset{x/x\#y\#}{x/y\in A}
  &\coqlink[gamma2]{\N{\gamma_2A}}&:=\mset{y/x\#y\#}{x/y\in A}
\end{align*}
and a function~\coqlink[gamma]{\emph{$\gamma A$}} defined as follows:
\begin{align*}
  \gamma\nil&:=\nil   \\ \gamma(x/y::A)&:=(\gamma A)x\#y\#
\end{align*}

\begin{lemma}[][sigma_gamma1]\label{lem-CFI-1}
  $\sigma_\#(\gamma_1 A)=A^1\#\,(\gamma A)$ and
  $\sigma_\#(\gamma_2 A)=A^2\#\,(\gamma A)$.
\end{lemma}
\begin{proof}
  By induction on $A$. 
\qed\end{proof}

\begin{lemma}[][gamma1_spec]\label{lem-CFI-2}
  Let $B\incl\gamma_i\,C$.
  Then there exists $A\incl C$ such that $\gamma_i\,A=B$.
\end{lemma}
\begin{proof}
  By induction on $B$ using Fact~\ref{fact-prelim-split}. 
\qed\end{proof}

\begin{lemma}[][gamma_inj]\label{lem-CFI-3}
  Let $\#$ not occur in $A_1$ and $A_2$.
  Then $\gamma A_1=\gamma A_2$ implies $A_1=A_2$.
\end{lemma}
\begin{proof} 
  By induction on $A_1$ using Fact~\ref{fact-prelim-split}.
\qed\end{proof}

\begin{theorem}[][reduction]
  PCP reduces to CFI.
\end{theorem}
\begin{proof}
  Let $P$ be a list of cards.
  We fix a symbol $\#$ not occurring in $P$
  and define $R_1:=\gamma_1P$ and $R_2:=\gamma_2P$.
  We show $\M{PCP}\,(P)\toot\M{CFI}\,(R_1,R_2,\#)$.

  Let $A\incl P$ be a nonempty match.
  Then $\gamma_1A\incl R_1$, $\gamma_2A\incl R_2$,
  and $\sigma_\#(\gamma_1A)=\sigma_\#(\gamma_2A)$
  by Lemma~\ref{lem-CFI-1}.

  Let $B_1\incl R_1$ and $B_2\incl R_2$
  be nonempty lists such that
  $\sigma_\#B_1=\sigma_\#B_2$.
  By Lemma~\ref{lem-CFI-2}
  there exist nonempty stacks $A_1,A_2\incl P$
  such that $\gamma_i(A_i)=B_i$.
  By Lemma~\ref{lem-CFI-1} we have
  $A_1^1\#(\gamma A_1)=A_2^2\#(\gamma A_2)$.
  By Fact~\ref{fact-prelim-split} we have
  $A_1^1=A_2^2$ and $\gamma A_1=\gamma A_2$.
  Thus $A_1=A_2$ by Lemma~\ref{lem-CFI-3}.
  Hence $A_1\incl P$ is a nonempty match.
\qed\end{proof}

Hopcroft et al.~\cite{hopcroft} give a reduction of PCP to CFI by using
grammars equivalent to the following Post grammars:
\begin{align*}
  {\N{\gamma_1A}}&:=\mset{x/i}{x/y\in A \text{ at position } i}
  &{\N{\gamma_2A}}&:=\mset{y/i}{x/y\in A \text{ at position } i}
\end{align*}
While being in line with the presentation of PCP with indices, it
complicates both the formal definition and the verification.

Hesselink~\cite{hesselink} directly reduces CFP to CFI for general
context-free grammars, making the reduction PCP to CFI redundant.
The idea is that a context-free grammar over $\Sigma$ contains a palindrome if and only if its
intersection with the context-free grammar of all palindromes over
$\Sigma$ is non-empty.
\setCoqFilename{CFP_CFI}
We give a \coqlink[CFP_CFI]{formal proof of this statement} using a definition
of context-free rewriting with explicit alphabets.

For Post grammars, CFP is not reducible to CFI, because the language
of all palindromes is not expressible by a Post grammar.

\section{TM to SRH}
\label{sec:TM-SRH}
\setCoqFilename{TM_SRH}
\newcommand{\none}{\ensuremath{\bot}}
\newcommand{\ltape}{\ensuremath{\llparenthesis}} 
\newcommand{\rtape}{\ensuremath{\rrparenthesis}} 

A Turing machine, independent from its concrete type-theoretic
definition, always consists of an alphabet $\Sigma$, a finite
collection of states $Q$, an initial state~$q_0$, a collection of
halting states $H \subseteq Q$, and a step function which controls the
behaviour of the head on the tape. The halting problem for Turing
machines TM then asks whether a Turing machine $M$ reaches a final
state when executed on a tape containing a string $x$.

In this section, we briefly report on our formalisation of a reduction
from TM to SRH following ideas from Hopcroft et al.~\cite{hopcroft}.
In contrast to the other sections, we omit the technical details
of the proof, because there are abundantly many, and none of them is
interesting from a mathematical standpoint.  We refer the interested
reader to~\cite{edith} for all details.

In the development, we use a formal definition of Turing machines
from Asperti and Ricciotti~\cite{aspertimulti}.

To reduce TM to SRH, a representation of configurations $c$ of Turing
machines as strings $\langle c \rangle$ is needed. Although the
content of a tape can get arbitrarily big over the run of a machine,
it is finite in every single configuration. It thus suffices to
represent only the part of the tape that the machine has previously
written to.

We write the current state to the left of the currently read symbol and,
following~\cite{aspertimulti}, distinguish four non-overlapping
situations: The tape is empty ($q\ltape \rtape$), the tape contains
symbols and the head reads one of them ($\ltape x qa y\rtape $), the
tape contains symbols and the head reads none of them, because it is
in a left-overflow position where no symbol has been written before
($q\ltape a x\rtape $) or the right-overflow counterpart of the latter
situation ($\ltape xaq\rtape$).  Note the usage of left and right
markers to indicate the end of the previously written part.

The reduction from TM to SRH now works in three steps. Given a Turing
machine $M$, one can define whether a configuration $c'$ is reachable
from a configuration $c$ using its transition
function~\cite{aspertimulti,edith}. First, we translate the transition
function of the Turing machine into a string rewriting system using
the translation scheme depicted in Table~\ref{table-TM}.

\aboverulesep=0ex
\belowrulesep=0ex
\renewcommand{\arraystretch}{1.1}
\begin{table}[H]
    \caption{Rewriting rules $x/y$ in $R$ if the machine according to
    its transition function in state $q_1$ continues in $q_2$ and
    reads, writes and moves as indicated. For example, if the
    transition function of the machine indicates that in state $q_1$ if
    symbol $a$ is read, the machine proceeds to state $q_2$, writes
    nothing and moves to the left, we add the rule $\ltape q_1a \, / \,
    q_2\ltape a$ and rules $c q_1 a \,/\, q_2 c a $ for every $c$ in
    the alphabet.}
  \label{table-TM}
  \centering\arrayrulecolor{gray}
  \begin{tabular}{l |l |l |l |l |l |l |l |l} \midrule
    Read\; & Write\; & Move\; & $x$ & $y$ & $x$ & $y$ & $x$ & $y$ \\ \arrayrulecolor{gray}\toprule

    $\bot$ & $\bot$ & $L$ & $q_1\ltape$ & $q_2\ltape$ & $a\,q_1\rtape$ & $q_2\,a\rtape$ & & \\ \midrule
    $\bot$ & $\bot$ & $N$ & $q_1\ltape$ & $q_2\ltape$ & $q_1\rtape$ & $q_2\rtape$ & & \\ \midrule
    $\bot$ & $\bot$ & $R$ & $q_1\ltape\rtape$ & $q_2\ltape\rtape$ & $q_1\rtape$ & $q_2\rtape$ &$q_1\ltape \,a$ & $\ltape\,q_1\,a$ \\ \midrule
    
    $\bot$ & $\some b$ & $L$ & $q_1\ltape$ & $q_2\ltape b$ & $a\,q_1\rtape$ & $q_2\,a\,b\rtape$ & & \\ \midrule
    $\bot$ & $\some b$ & $N$ & $q_1\ltape$ & $\ltape q_2\,b$ &  $q_1\rtape$ & $q_2\,b\rtape$ & & \\ \midrule
    $\bot$ & $\some b$ & $R$ & $q_1\ltape$ & $\ltape b\,q_2$ &  $q_1\rtape$ & $b\,q_2\rtape$ & & \\ \midrule
    
    $\some a$ & $\bot$ & $L$ & $\ltape q_1\,a$ & $q_2\ltape a$ & $c\,q_1\,a$ & $q_2\,c\,a$ & & \\ \midrule
    $\some a$ & $\bot$ & $N$ & $q_1\,a$ & $q_2\,a$ & & & & \\ \midrule
    $\some a$ & $\bot$ & $R$ & $q_1\,a$ & $a\,q_2$ & & & & \\ \midrule
    
    $\some a$ & $\some b$ & $L$ & $\ltape q_1\,a$ & $q_2\ltape b$ & $c\,q_1\,a$ & $q_2\,c\,b$ & & \\ \midrule
    $\some a$ & $\some b$ & $N$ & $q_1\,a$ & $q_2\,b$ & & & & \\ \midrule
    $\some a$ & $\some b$ & $R$ & $q_1\,a$ & $b\,q_2$ & & & & \\ \midrule
  \end{tabular}
  \vspace{1mm}
\end{table}

\begin{lemma}[][reduction_reach_]
  For all Turing machines $M$ and configurations $c$ and $c'$ there is
  a SRS $R$ such that
  $\langle c \rangle \succ^*_R \langle c' \rangle$ if and only if the
  configuration $c'$ is reachable from the configuration $c$ by the
  machine $M$.
\end{lemma}

In the development, we first reduce to a version of string rewriting
with explicit alphabets, and then \coqlink[reduction]{reduce this version} to string
rewriting as defined before.

This proof is by far the longest in our development. In its essence,
it is only a shift of representation, making explicit that transition
functions encode a rewriting relation on configurations. The proof
is mainly a big case distinction over all possible shapes of
configurations of a machine, which leads to a combinatorial explosion
and a vast amount of subcases.
The proof does, however, not contain any surprises or insights.

Note that, although we work with deterministic machines in the Coq
development, the translation scheme described in Table~\ref{table-TM}
also works for nondeterministic Turing machines.

The second step of the reduction is to incorporate the set of halting
states~$H$.  We define an intermediate problem $\M{SRH}'$,
generalising the definition of $\M{SRH}$ to strings:

\[\M{SRH}'(R, x, z) := \exists y.~x \succ^*_R y \land \exists a \in
  z.~a \in y \]

Note that $\M{SRH}(R,x,a) \toot \M{SRH}'(R,x,[a])$.
TM can then easily be reduced to $\M{SRH}'$:

\begin{lemma}[][halt_SRH']
  $\M{TM}$ reduces to $\M{SRH}'$.
\end{lemma}
\begin{proof}
  Given a Turing machine $M$ and a string $x$, $M$ accepts $x$ if and
  only if $\M{SRH}(R,q_0\ltape x \rtape, z)$, where $R$ is the system
  from the last lemma, $q_0$ is the starting state of $M$ and $z$ is a
  string containing exactly all halting states of $M$.
\qed\end{proof}

Third, we can reduce $\M{SRH}'$ to $\M{SRH}$:

\begin{lemma}[][SRH'_SRH]
  $\M{SRH}'$ reduces to $\M{SRH}$.
\end{lemma}
\begin{proof}
  Given a SRS $R$, a string $x$ and a string $z$, we first
  fix an alphabet $\Sigma$ covering $R$ and $x$, and a fresh symbol
  $\#$. We then have $\M{SRH'}(R,x,z)$ if and only if
  $\M{SRH}(R \app \left[ a/\# \mid a \in z\right],x,\#)$.
\qed\end{proof}

All three steps combined yield:

\begin{theorem}[][Halt_SRH]
  $\M{TM}$ reduces to $\M{SRH}$.
\end{theorem}

\section{Discussion}

We have formalised and verified a number of computational reductions
to and from the Post correspondence problem based on Coq's type
theory.  Our goal was to come up with a development as elegant as
possible.  Realising the design presented in this paper in Coq yields
an interesting exercise practising the verification of list-processing
functions.  If the intermediate lemmas are hidden and just the
reductions and accompanying correctness statements are given, the
exercise gains difficulty since the correctness proofs for the
reductions $\M{SR}\preceq\M{MPCP}\preceq\M{PCP}$ require the invention
of general enough inductive invariants
(Lemmas~\ref{lem-SRMPCP-trans1}, \ref{lem-SRMPCP-trans2},
\ref{lem-MPCP-PCP-trans1},~\ref{lem-MPCP-PCP-trans2}).  To our
surprise, we could not find rigorous correctness proofs for the
reductions $\M{TM}\preceq\M{SR}\preceq\M{MPCP}\preceq\M{PCP}$ in the
literature (e.g,~\cite{hopcroft,davis,sipser}).  Teaching these
reductions without rigorous correctness proofs in theoretical computer
science classes seems bad practice.  As the paper shows, elegant and
rigorous correctness proofs using techniques generally applicable in
program verification are available.

The ideas for the reductions
$\M{TM} \preceq \M{SRH}\preceq\M{SR}\preceq\M{MPCP}\preceq\M{PCP}$
are taken from Hopcroft et al.~\cite{hopcroft}.
They give a monolithic reduction of
the halting problem for Turing machines to MPCP.
The decomposition $\M{TM}\preceq\M{SRH}\preceq\M{SR}\preceq\M{MPCP}$
is novel.
Davis et al.~\cite{davis} give
a monolithic reduction $\M{SR}\preceq\M{PCP}$
based on different ideas.
The idea for the reduction $\M{PCP}\preceq\M{CFP}$
is from Hesselink~\cite{hesselink},
and the idea for the reduction $\M{PCP}\preceq\M{CFI}$
appears in Hopcroft et al.~\cite{hopcroft}.

There are several design choices
we faced when formalising the material
presented in this paper.
\begin{enumerate}
\item We decided to formalise PCP without
  making use of the positions of the cards in the list $P$.
  Most presentations in the literature
  (e.g.,~\cite{hopcroft,sipser})
  follow Post's original paper~\cite{emilpost}
  in using positions (i.e., indices) rather than cards in matches.
  An exception is Davis et al.~\cite{davis}.
  We think formulating PCP with positions
  is an unnecessary complication.
\item We decided to represent symbols as numbers rather
  than elements of finite types serving as alphabets.
  Working with implicit alphabets represented as lists
  rather than explicit alphabets represented as
  finite types saves bureaucracy.
\item We decided to work with Post grammars (inspired by
  Hesselink~\cite{hesselink}) rather than general context-free
  grammars since Post grammars sharpen the result and enjoy a
  particularly simple formalisation.  In the Coq development, we show
  that \setCoqFilename{Post_CFG}\coqlink[reduce_grammars]{Post
    grammars are an instance of context-free grammars}.
\end{enumerate}

Furthermore, we decided to put the focus of this paper on the elegant
reductions and not to cover Turing machines in detail. While being a
wide-spread model of computation, even the concrete formal definition
of Turing machines contains dozens of details, all of them not
interesting from a mathematical perspective.

The Coq development
verifying the results of sections 3 to 7
consists of about 850 lines
of which about one third realises specifications.
The reduction
$\M{SR}\preceq\M{SRH}$
takes 70 lines,
${\M{SR}\preceq\M{MPCP}}$
takes 105 lines,
$\M{MPCP}\preceq\M{PCP}$
takes 206 lines,
$\M{PCP}\preceq\M{CFP}$
takes 60 lines,
and $\M{PCP}\preceq\M{CFI}$
takes 107 lines.
\setCoqFilename{singleTM}
The reduction $\M{TM} \preceq \M{SRH}$ takes 610 lines, 230 of them
specification, plus a definition of Turing machines taking 291 lines.

\subsection*{Future Work}

Undecidability proofs for logics are often done by reductions from PCP
or related tiling problems.  We thus want to use our work as a
stepping stone to build a library of reductions which can be used to
verify more undecidability proofs.  We want to reduce PCP to the
halting problem of Minsky machines to prove the undecidability of
intuitionistic linear logic~\cite{DLW}.  Another possible step would
be to reduce PCP to validity for first-order logic~\cite{Church},
following the reduction from e.g.~\cite{manna}.  Many other
undecidability proofs are also done by direct reductions from PCP,
like the intersection problem for two-way-automata~\cite{RabinScott},
unification in third-order logic~\cite{Huet}, typability in the
$\lambda\Pi$-calculus~\cite{Dowek}, satisfiability for more applied
logics like HyperLTL~\cite{hahn}, or decision problems of first order
theories~\cite{Treinen}.

In this paper, we gave reductions directly as functions in Coq instead
of appealing to a concrete model of computation.  Writing down
concrete Turing machines computing the reductions is possible in
principle, but would be very tedious and distract from the elegant
arguments our proofs are based on.

In previous work~\cite{forster} we studied an explicit model of
computation based on a weak call-by-value calculus L in Coq.
L would allow an implementation of all reduction functions
without much overhead, which would also formally establish the
computability of all reductions. 

Moreover, it should be straightforward to reduce PCP to the
termination problem for L. Reducing the termination problem of L to TM
would take considerable effort. Together, the two reductions would
close the loop and verify the computational equivalence of TM, SRH,
SR, PCP, and the termination problem for L. Both reducing PCP to L and
implementing all reductions in L is an exercise in the verification of
deeply embedded functional programs, and orthogonal in the necessary
methods to the work presented in this paper.
 
\bibliographystyle{plain}
\bibliography{references_clean}

\begin{thebibliography}{10}

\bibitem{aspertimulti}
Andrea Asperti and Wilmer Ricciotti.
\newblock A formalization of multi-tape {T}uring machines.
\newblock {\em Theoretical Computer Science}, 603:23--42, 2015.

\bibitem{Church}
Alonzo Church.
\newblock A note on the {E}ntscheidungsproblem.
\newblock {\em J. Symb. Log.}, 1(1):40--41, 1936.

\bibitem{davis}
Martin~D. Davis, Ron Sigal, and Elaine~J. Weyuker.
\newblock {\em Computability, Complexity, and Languages: Fundamentals of
  Theoretical Computer Science}.
\newblock Academic Press, 2nd edition, 1994.

\bibitem{Dowek}
Gilles Dowek.
\newblock The undecidability of typability in the lambda-pi-calculus.
\newblock In {\em International Conference on Typed Lambda Calculi and
  Applications}, pages 139--145. Springer, 1993.

\bibitem{hahn}
Bernd Finkbeiner and Christopher Hahn.
\newblock Deciding hyperproperties.
\newblock In {\em {CONCUR} 2016}, pages 13:1--13:14, 2016.

\bibitem{forster}
Yannick Forster and Gert Smolka.
\newblock Weak call-by-value lambda calculus as a model of computation in
  {C}oq.
\newblock In {\em {ITP} 2017}, pages 189--206. Springer, LNCS 10499, 2017.

\bibitem{edith}
Edith Heiter.
\newblock {U}ndecidability of the {P}ost correspondence problem in {C}oq.
\newblock Bachelor's Thesis, Saarland University,
  \url{https://www.ps.uni-saarland.de/~heiter/bachelor.php}, 2017.

\bibitem{hesselink}
Wim~H. Hesselink.
\newblock {P}ost's correspondence problem and the undecidability of
  context-free intersection.
\newblock Manuscript, University of Groningen,
  \url{http://wimhesselink.nl/pub/whh513.pdf}, 2015.

\bibitem{hopcroft}
John~E. Hopcroft, Rajeev Motwani, and Jeffrey~D. Ullman.
\newblock {\em Introduction to Automata Theory, Languages, and Computation}.
\newblock Addison-Wesley, 3rd edition, 2006.

\bibitem{Huet}
Gerard~P. Huet.
\newblock The undecidability of unification in third order logic.
\newblock {\em Information and control}, 22(3):257--267, 1973.

\bibitem{DLW}
Dominique Larchey-Wendling and Didier Galmiche.
\newblock The undecidability of boolean {BI} through phase semantics.
\newblock In {\em LICS 2010}, pages 140--149. IEEE, 2010.

\bibitem{manna}
Zohar Manna.
\newblock {\em Mathematical theory of computation}.
\newblock Dover Publications, Incorporated, 2003.

\bibitem{emilpost}
Emil~L. Post.
\newblock A variant of a recursively unsolvable problem.
\newblock {\em Bulletin of the American Mathematical Society}, 52(4):264--268,
  1946.

\bibitem{RabinScott}
Michael~O. Rabin and Dana Scott.
\newblock Finite automata and their decision problems.
\newblock {\em IBM journal of research and development}, 3(2):114--125, 1959.

\bibitem{sipser}
Michael Sipser.
\newblock {\em Introduction to the Theory of Computation}.
\newblock Cengage Learning, {I}nternational edition, 2012.

\bibitem{Coq}
{The Coq Proof Assistant}.
\newblock \url{http://coq.inria.fr}, 2017.

\bibitem{Treinen}
Ralf Treinen.
\newblock A new method for undecidability proofs of first order theories.
\newblock {\em Journal of Symbolic Computation}, 14(5):437--457, 1992.

\end{thebibliography}

\end{document}